\documentclass[journal, draftcls, onecolumn]{IEEEtran}
\usepackage{amssymb,amsmath,graphicx,cite,graphics,mathrsfs}
\usepackage{latexsym,bbm,float, mathtools}
\usepackage{amsthm}

\usepackage{cleveref}
\usepackage{lipsum}
\usepackage{amsfonts}
\usepackage{graphicx}
\usepackage{epstopdf}
\usepackage{algorithmic}
\usepackage{multirow}
\ifpdf
  \DeclareGraphicsExtensions{.eps,.pdf,.png,.jpg}
\else
  \DeclareGraphicsExtensions{.eps}
\fi

\newtheorem{proposition}{Proposition}
\newtheorem{theorem}{Theorem}
\newtheorem{corollary}{Corollary}
\newtheorem{lemma}{Lemma}
\newtheorem{definition}{Definition}

\theoremstyle{remark}

\newcommand{\supp}{\textnormal{supp}}

\begin{document} \date{}
\allowdisplaybreaks
\IEEEoverridecommandlockouts

\title{Sparse Combinatorial Group Testing}

\author{\IEEEauthorblockN{Huseyin A. Inan, Peter Kairouz, and Ayfer Ozgur \thanks{This work was presented in part at Allerton 2017. This work was supported in part by NSF grant \#1514538.} } }

\maketitle

\begin{abstract}
In combinatorial group testing, the primary objective is to fully identify the set of at most $d$ defective items from a pool of $n$ items using as few tests as possible. The celebrated result for the combinatorial group testing problem is that the number of tests $t$ can be made logarithmic in $n$ when $d = O( \textnormal{poly}(\log n) )$. However, state-of-the-art group testing codes require the items to be tested $w = \Omega\left(d \log n\right)$ times and tests to include $\rho = \Omega\left(\frac{n}{d}\right)$ items (within logarithmic factors). In many emerging applications, items can only participate in a limited number of tests and tests are constrained to include a limited number of items.

In this paper, we study the ``sparse'' regime for the group testing problem where we restrict the number of tests each item can participate in by $w_{\max}$ or the number of items each test can include by $\rho_{\max}$ in both noiseless and noisy settings. These constraints lead to a largely unexplored regime where $t$ is a fractional power of $n$, rather than logarithmic in $n$ as in the classical setting. Our results characterize the number of tests $t$ needed in this regime as a function of $w_{\max}$ or $\rho_{\max}$ and show, for example, that $t$ decreases drastically when $w_{\max}$ is increased beyond a bare minimum. In particular, it is easy to see in the noiseless case that if $w_{\max} \leq d$, then we must have $t=n$, i.e., testing every item individually is optimal. We show that if $w_{\max}=d+1$, the number of tests decreases suddenly from $t=n$ to $t = \Theta(d \sqrt{n})$. The order-optimal construction is obtained via a modification of the classical Kautz-Singleton construction, which is known to be suboptimal for the classical group testing problem. For the more general case, when $w_{\max}=ld+1$ for integer $l>1$, the modified Kautz-Singleton construction requires  $t = \Theta(d n^{\frac{1}{l+1}})$ tests, which we prove to be near order-optimal. We also show that our constructions have a favorable encoding and decoding complexity, i.e. they can be decoded in  $(\textnormal{poly}(d) + O(t))$-time and each entry in any codeword can be computed in space $\textnormal{poly}(\log n)$. We finally discuss an application of our results to the construction of energy-limited random access schemes for IoT networks, which provided the initial motivation for our work.
\end{abstract}

\section{Introduction}

Group testing is a subfield of combinatorial mathematics that studies how to identify a set of $d$ (or less) defective items from a large population of size $n$. For an unknown sequence ${x} \in \{0,1\}^n$ with at most $d$ ones representing the defective items, we are allowed to test any subset $S \subseteq \{1, \ldots, n\}$ of the items. The result of a test $S$ could either be positive, which happens when at least one item in $S$ is defective (i.e., $\exists~ i \in S$ such that $x_i = 1$), or negative when all the items in $S$ are not defective (i.e., $\forall i \in S$ we have $x_i = 0$). The goal is to design as few tests as possible (say $t$ tests) so that we can exactly recover the unknown sequence ${x}$.

The original group testing framework was developed in 1943 by Robert Dorfman \cite{dorfman1943detection}. Back then, group testing was devised to identify which WWII draftees were infected with syphilis -- without having to test them individually. In Dorfman's application, items represented draftees and tests represented actual blood tests. Over the years, group testing has found numerous applications in an array of exciting fields spanning biology \cite{chen2008decoding}, medicine \cite{ganesan2017learning}, machine learning \cite{malioutov2013exact}, data analysis \cite{gilbert2008group}, computer science \cite{goodrich2005indexing}, and signal processing \cite{emad2014poisson}. In addition, group testing has been extensively applied to various disciplines of wireless communication, such as multiple access control protocols \cite{kautz1964,Berger1984,wolf1985born,Fletcher2009} and neighborhood discovery \cite{luo2008neighbor}.

The celebrated result for the group testing problem is that $t$ can be made logarithmic in $n$. One of the earliest explicit group testing constructions is due to Kautz and Singleton and requires $t = O(d^2 \log_d^2 n)$ tests \cite{kautz1964}.  This construction uses a Reed-Solomon code concatenated with a non-linear identity code and it matches the best known lower bound $\Omega(d^2 \log_d n)$ \cite{d1982bounds, furedi1996onr} in the regime where $d = \Theta(n^{\alpha})$ for some $\alpha \in (0, 1)$.  More recently, a different explicit construction achieving $t = O(d^2 \log n)$ was introduced by Porat and Rothschild in \cite{porat2008}, which outperforms Kautz and Singleton's construction in the regime  where $d = O(\textnormal{poly}(\log n))$. 

These results imply that group testing can provide drastic gains when $d\ll n$, say $d = O(\textnormal{poly}(\log n))$, compared to the naive approach of testing every item individually which results in $t=n$ total number of tests. However, it can be shown that these constructions require each item to participate in $w = \Omega\left(\frac{d \log n}{\log d + \log \log n} \right)$ tests and each test to include $\rho = \Omega\left(\frac{n}{d \log_d n}\right)$ items. In many applications, the total number of tests that can be performed on each item or the number of items each test can include can be limited due to different reasons. For example, the amount of blood or genetic material available from an individual can limit the number of tests that this individual can participate in. Similarly, equipment limitations and testing procedures can impose a maximum on the number of samples that can be simultaneously tested. For example, combining too many blood samples in one test often increases the misdetection probability of the targeted disease. Our interest in the group testing problem was mainly motivated by its applications in wireless communication, in particular the use of group testing codes for constructing random access schemes for the Internet of Things (IoT) networks, as we discuss in more detail in the last section of the paper. Here items represent sensor nodes and tests represent binary transmissions over a common channel, and the energy constraint at the physical layer can be translated to a constraint on the number of tests applied to each item in the group testing framework.

Motivated by these observations, in this paper we study the group testing problem when the number of tests each item can participate in is restricted by $w_{\max}$ or when the number of items each test can include is restricted by $\rho_{\max}$. In the noiseless case, it is not difficult to show that if $w_{\max} \leq d$ or if $\rho_{\max}\leq d + 1$, then we need  $t=n$ tests, i.e., testing every item individually is optimal. A natural question then is: how does $t$ decrease as we increase $w_{\max}$ and $\rho_{\max}$ beyond these bare minimums (up to their values in state-of-the-art constructions)? In particular, can we slightly increase $w_{\max}$ and $\rho_{\max}$ beyond $d$ and $d+1$, respectively, and significantly reduce $t$, the number of tests needed? The answer turns out to be positive when we have an $w_{\max}$ constraint but not so for the case with a $\rho_{\max}$ constraint.

We show that when $w_{\max}=d+1$, the number of tests decreases drastically from $t=n$ to $t = (d+1)\sqrt{n}$. More generally, if $w_{\max}=ld+1$ for any positive integer $l$ such that $ld+1 \leq
\sqrt[l+1]{n}$, we can achieve
$$t=(ld+1)n^{\frac{1}{l+1}}.$$
This implies that the fractional power of $n$ can be reduced drastically when $w_{\max}$ is increased as a multiple of $d$. Note that this result is most significant when $d = O(\textnormal{poly}(\log n))$. We achieve this performance by introducing a simple modification of Kautz and Singleton's construction, which shows that the field size in this construction can be used to trade between $t$ and $w_{max}$. We then prove a nearly matching lower bound which shows that $$t = \Omega(d^{\frac{2}{l+1}}n^{\frac{1}{l+1}}).$$
In particular when $w_{\max} = d + 1$, this shows that Kautz and Singleton's construction is  order-optimal (to an almost matching constant). This is somewhat surprising given that Kautz and Singleton's construction is strictly suboptimal in the classical group testing setting when $d = O(\textnormal{poly}(\log n))$. 

As opposed to the case with an $w_{\max}$ constraint, the decrease in $t$ is much less dramatic with increasing $\rho_{max}$. We prove that
for any $\rho_{\max}$,
\begin{align*}
t \geq \dfrac{(d+1) n}{\rho_{\max}}
\end{align*}
in the noiseless case. In other words, $\rho_{\max}$ needs to be a fractional power of $n$, in order for $t$ to scale sublinearly in $n$. When $\rho_{\max} = \Theta(n^{\alpha})$, we prove that there exists a construction with $t = \Theta(d n^{1-\alpha})$. We also show that when $\alpha = l/(l+1)$ for any fixed integer $l \geq 1$, Kautz and Singleton's construction can be used to achieve the order-optimal $t = (l d + 1) \sqrt[l+1]{n}$ tests. For the special case $l=1$, it is interesting to note that this constructions matches the lower bound with exact constants, therefore, the Kautz-Singleton construction is exactly optimal. We further extend these results to the noisy case  when a certain number of tests can have faulty outcomes. We also prove that our constructions can be decoded in  $(\textnormal{poly}(d) + O(t))$-time and each entry in any codeword can be computed in space $\textnormal{poly}(\log n)$. This shows that these constructions not only (nearly) achieve the fundamental lower bounds, but also have a favorable encoding and decoding complexity.

\subsection{Comparison with Prior Work}

To the best of our knowledge, our problem formulation is novel and has not been widely explored in the group testing literature. The only exception is a recent paper by Gandikota et al.  \cite{gandikota2016sparse}. However, Gandikota et al. focus (for the most part) on statistical approaches that provide lower and upper bounds on the number of tests when the defective set can be recovered with $\epsilon$-error for some arbitrarily small but positive $\epsilon>0$, while our approach is purely combinatorial, i.e. we aim to recover the defective set with $0$-error. Gandikota et al use information theoretic techniques based on Fano's inequality to prove lower bounds for the $\epsilon$-error case and provide randomized constructions, while our lower bounds are combinatorial and our constructions are explicit. For explicit constructions, their paper refers to an earlier paper by Macula \cite{macula1996constant}, however, the construction in \cite{macula1996constant} is highly suboptimal as we discuss next.

\begin{table}[h]

\begin{center}
\begin{tabular}{|p{3cm}|c|c|c|c|}
\hline
\multirow{2}{*}{Setting} & \multicolumn{2}{ |c| }{This paper} & \multicolumn{2}{ |c| }{Gandikota et al.  \cite{gandikota2016sparse}}  \\ \cline{2-5}
& Lower Bound & Upper Bound & Lower Bound & Upper Bound  \\ \cline{1-5}
\multirow{3}{3cm}{ }  & & & & \\ Sparse Codewords  with $w = ld+1$ &
$\Omega\left( d^{\frac{2}{l+1}} \sqrt[l+1]{n}\right) $
& $O\left( d \sqrt[l+1]{n}\right) $ & $\Omega\left( d^2 (\frac{n}{d})^{\frac{1}{ld+1}}\right) $ & $O\left( d^2 (n \left( \frac{n}{d})^d\right) ^{\frac{1}{ld+1}}\right) $
 \\ & & (Explicit) & & (Randomized) \\ \hline
\multirow{3}{3cm}{ }  & & & &    \\
 Sparse Tests  with $\rho = n^{\frac{l}{l+1}}$ & $\Omega\left( d \sqrt[l+1]{n}\right) $
 & $O\left( d \sqrt[l+1]{n}\right) $
 & $\Omega\left( \sqrt[l+1]{n} \frac{\log \left( n/d \right) }{\log \left( \sqrt[l+1]{n}/d \right) } \right)$
 & $ O\left( \sqrt[l+1]{n} \log \left( n ( \frac{n}{d} )^d \right)  \right)$
  \\ & & (Explicit) & & (Randomized)\\ \hline
\end{tabular}
\end{center}
\vspace{3pt}
\caption{Comparison of sparse group testing results.}
\label{t1}
\end{table}

Table \ref{t1} compares the results provided in this paper with the ones presented in \cite{gandikota2016sparse}.
When the weights of the columns are constrained with $w_{\max} = ld+1$, the dependence  on $n$ in our lower bound is $\sqrt[l+1]{n}$ whereas \cite{gandikota2016sparse} provides $\sqrt[ld+1]{n}$ which is significantly weaker. In terms of upper bounds, we provide an explicit construction achieving $O\left( d \sqrt[l+1]{n}\right)$ which is better than the randomized construction in  \cite{gandikota2016sparse}. Note for example that when $l=1$ our upper bound gives $d\sqrt{n}$ while their upper bound gives $d^{1+\frac{2}{d+1}}n$. The explicit construction \cite{macula1996constant}, referred to in \cite{gandikota2016sparse}, provides an upper bound of $O( n^{\frac{1}{(ld+1)^{1/d}} })$ which is substantially weaker than the results in this paper.

When the weights of the rows are constrained to be less than or equal to $\rho_{\max} = n^{\frac{l}{l+1}}$, one can observe that the lower bound presented in this paper achieves $\Omega\left( d \sqrt[l+1]{n}\right)$, which is better by a factor of $d$ when compared to the $\Omega\left(\sqrt[l+1]{n}\right)$ lower bound provided in \cite{gandikota2016sparse}. In terms of the upper bound, we provide an explicit construction that achieves
$O\left( d \sqrt[l+1]{n}\right)$ which matches our lower bound. On the other hand, the randomized construction in \cite{gandikota2016sparse} is off by a factor of $\log(n/d)$. The explicit construction \cite{macula1996constant} referred to in \cite{gandikota2016sparse} provides an upper bound of $O\left( \sqrt[l+1]{n} (l+1)^d  \right) $ which has an exponential term in $d$. More recently, we presented a sparse group testing framework for the stochastic setting in \cite{inan18isit}.

\subsection{Paper Organization}
The remainder of this paper is organized as follows. In Section \ref{sec:prior}, we present the needed prerequisite material and describe two common combinatorial group testing constructions. 
The main results of our paper are formally presented and proved in Section \ref{sec:mainres}. In Section \ref{sec:decoding}, we discuss the encoding and decoding complexities of the explicit constructions we present in Section \ref{sec:mainres}. In Section \ref{sec:application}, we discuss an application of our results to wireless random access which provided the original motivaiton for our work. 
Finally, we conclude our paper in Section \ref{sec:conclusion} by noting a few interesting and nontrivial extensions.

\section{Preliminaries}
\label{sec:prior}
For a $t \times n$ binary matrix ${M}$, we use ${M}_j$ to refer to its $j$'th column and ${M}_{ij}$ to refer to its $(i, j)$'th entry. For an integer $m \geq 1$, we denote the set $\{1, \ldots, m\}$ by $[m]$. The support of a column $M_j$ is denoted as $\textnormal{supp}(M_j) \coloneqq \{ i :  {M}_{ij} = 1\}$. We say that a binary column ${M}_i$ \textit{covers} a binary column ${M}_j$ if  $M_i \vee M_j = M_i$, or equivalently $\supp(M_j) \subseteq \supp(M_i)$. The Hamming weight of a row or a column of ${M}$ will be simply referred to as the \textit{weight} of the row or column and $w_j$ represents the weight of $j$'th column.


\subsection{Non-adaptive Combinatorial Group Testing}

Our paper focuses on non-adaptive combinatorial group testing (CGT). Non-adaptive refers to the fact that the tests are designed and fixed a priori, in constrast to the adaptive case, where the tests are designed sequentially, meaning that the $j^{th}$ test is a function of the outcomes of the $j-1$ previous tests. Combinatorial refers to the fact that we want our group testing schemes to  recover the set of defective items with zero-error, in contrast to the probabilistic approach which allows for a  small probability of error. A non-adaptive CGT strategy can be represented by a $t \times n$ binary matrix ${M}$, where $M_{ij} = 1$ indicates that item $j$ participates in test $i$. We will occasionally refer to $M$ as a group testing code (or codebook) and its $i^{th}$ column $M_i$ as the $i^{th}$ codeword.  
A necessary condition for the design of a non-adaptive CGT strategy ${M}$ is that of \textit{separability}. A matrix ${M}$ is $d$-\textit{separable} if for any $x_1 \neq x_2$, $d$-sparse vectors,  we have that $Mx_1 \neq Mx_2$.
Unfortunately, the $d$-separability condition does not lead to tractable, explicit, and efficiently decodable constructions of $M$ for an arbitrary value of $n$. To circumvent this issue, a stronger condition on ${M}$ is needed. This condition is known as $d$-\textit{disjunctiveness} \cite{du2000combinatorial}.
We first revisit the definition of $d$-{disjunctiveness} \cite{du2000combinatorial}.
\begin{definition}
A $t \times n$ binary matrix $M$ is called $d$-\textit{disjunct} if any Boolean sum of up to $d$ columns of ${M}$ does not cover any other column not included in the sum.
\end{definition}
The $d$-disjunctiveness  property ensures that we can recover up to $d$ columns from their Boolean sum. This can be naively done using the \textit{cover decoder}. The cover decoder simply scans through the columns of ${M}$, and checks whether or not the test results vector $Y$ covers a particular column. If column $i$ is covered by $Y$, then item $i$ is declared defective. When $M$ is $d$-disjunct, the cover decoder succeeds at identifying all the defective items, while achieving a zero false positive rate. Interestingly, one can also show that $(d+1)$-separability implies $d$-disjunctiveness \cite{du2000combinatorial}. Therefore, even though disjunctiveness is stronger than separability, the two conditions are essentially equivalent.

We define $t(d, n)$ to be the smallest $t$ needed for a binary $t \times n$ matrix $M$ to be $d$-disjunct. Notice that naturally, $t(d, n) \leq n$ because we can always use the identity matrix ${M} = I_n$ to identify any $1 \leq d \leq n$ defectives among $n$ items. A classical result in the non-adaptive combinatorial group testing literature shows that $t(d,n) = \Omega({d^2}\log_d n)$ \cite{d1982bounds,furedi1996onr}. Several explicit and randomized constructions of $d$-disjunct matrices have been developed over the past 50 years with the most efficient (when $d = O(\textnormal{poly}(\log n))$) known constructions achieved $t=O(d^2\log n)$ \cite{porat2008, du2000combinatorial,alon2006algorithmic}.  

\subsection{Relevant Lower Bounds}
We now summarize two known lower bounds on the minimum number of tests. These bounds imply that individual testing is necessary whenever $d = \Omega(\sqrt{n})$ or $w_{\max} \leq d$, where $w_{\max}$ is the maximum number of tests an item participates in (or equivalently, the maximum column weight).

\begin{proposition}
\label{prop:larged}
For all $n$ and $d$, the following bound on $t(d, n)$ holds
\begin{align}
t(d, n) \geq \min \left\lbrace \dbinom{d+2}{2}, n
 \right\rbrace.
 \label{simpleLB}
\end{align}
\end{proposition}
Proposition \ref{prop:larged} suggests that we need $d = O(\sqrt{n})$ to be able to design a $d$-disjunct matrix with $t < n$. 

\begin{proposition}
\label{prop:smallw}
If $w_{\max} \leq d$, then $t(d,n) = n$.
\end{proposition}
The above proposition shows that one cannot do better than individual testing when the maximum number of tests an item can participate in is less than or equal to $d$. 

The proofs of these proposition are due to D'yachkov and Rykov, and can be found in \cite{d1982bounds, rykov}.

\subsection{Disjunct Matrices via Error Correcting Codes}
A $q$-nary error-correcting code is a code whose codewords consist of $q$ basic symbols \cite{eec}. Binary codes are a special case of $q$-nary codes with $q = 2$. Consider a $q$-nary code with $n = q^k$ codewords of length $t = k + r$. Denoting the minimum distance between the codewords as $d_{\min}$, one can show that $d_{\min} \leq r+1$ from the following observation. Fix any $k$ positions in the codewords. If any two codewords have the same symbols in these positions, then it must be the case that $d_{\min} \leq r$. Otherwise, we must observe all possible $q^k$ sequences in the $k$ fixed positions. In this case, some of the codewords will differ by only one position on the fixed $k$ positions. Hence, $d_{\min} \leq r+1$. We state this formally in the following theorem \cite{mds}.

\begin{theorem}
A $q$-nary code with $n = q^k$ codewords of length $t = k+r$ must satisfy $d_{\min} \leq r+1$.
\end{theorem}

Codes with $d_{\min} = r+1$ and $n = q^k$ are called maximum distance separable (MDS) codes \cite{mds}.  Reed-Solomon codes \cite{RS} are a known class of MDS codes with the constraint that $q \geq t$. When concatenated with a nonlinear code, Reed-Solomon codes lead to $d$-disjunct group testing codes. In what follows, we will use the subscript $q$ in the parameters of the Reed-Solomon codes to separate them from the group testing codes that will be constructed shortly. To recap, Reed-Solomon codes achieve a minimum distance of $d_{\min}= r_q + 1$ with a code length of $t_q = k_q + r_q$ and a number of codewords equal to
$n_q = q^{k_q}$, provided that $q \geq t_q$ and $q$ is a prime power.

We can convert a Reed-Solomon code into a group testing code using the following method introduced by Kautz and Singleton in
\cite{kautz1964}. We replace each codeword symbol $i \in \{ 1, 2, \ldots, q \}$ by $e_{i}$, a length-$q$ binary sequence with a single nonzero entry in the $i^{th}$ position. Thus, a Reed-Solomon code is transformed into a binary code of length $t = q t_q$ by concatenating it with the ``identity code''. The minimum distance of the resultant binary code is double that of the Reed-Solomon code; i.e., $d_{\min} = 2(r_q + 1)$. This is because any two distinct $q$-nary symbols will differ in two positions in their corresponding length $q$ binary sequences. Note that the number of codewords remains the same $n = n_q = q^{k_q}$,  and all the binary codewords have the same weight $w = t_q$. This construction will be referred to as the Kautz-Singleton construction.

Consider a binary code $M$ with minimum codeword weight of $w_{\min}$. We define $\lambda_{\max} $ to be the maximum number of overlapping ones between any two codewords in $M$. In the coding theory literature, $\lambda_{\max}$ is commonly referred to as the maximal correlation of $M$.
A central result in group testing demonstrates that $M$ is $d$-{disjunct} as long as $\lambda_{\max} d + 1 \leq w_{\min}$. This can be seen from the following simple argument. Take any $d+1$ codewords and fix one codeword among them. The number of overlapping ones between the fixed codeword and the rest of the codewords is at most $d \lambda_{\max}$. Since the minimum weight satisfies $w_{\min} \geq d \lambda_{\max} + 1$, this codeword cannot be covered by the rest of the codewords. Thus, $M$ must be at least $d$-disjunct. We state this formally in the following lemma.
\begin{lemma}
A binary code $M$ with codewords of minimum weight $w_{\min}$ and maximal correlation $\lambda_{\max}$ is $\left\lfloor\frac{w_{\min}-1}{\lambda_{\max}} \right \rfloor$-disjunct.
\end{lemma}
Observe that in the Kautz-Singleton construction, we have that
\begin{align*}
\lambda_{\max} = w - d_{\min}/2 =  t_q - r_q - 1
= k_q - 1.
\end{align*}
Therefore, the Kautz-Singleton construction provides us with a group testing code
that is $\left\lfloor\frac{t_q-1}{k_q -1} \right \rfloor$-{disjunct}.

\begin{theorem}
The Kautz-Singleton construction provides a $t \times n$ $d$-{disjunct} matrix where $t = O(d^2 \log_d^2 n)$
with constant column weight $w = \Omega\left( \frac{d \log n}{\log d + \log \log n} \right)$ and constant row weight $\rho = \Omega\left( \frac{n}{d \log_d n}\right) $.
\end{theorem}

\begin{proof}
 To obtain a $d$-{disjunct} code using the Kautz-Singleton construction, we set  $t_q = q$, and choose $q$ and $k_q$ such that
$d = \left\lfloor\frac{q-1}{k_q -1} \right \rfloor$. Note that $n = q^{k_q}$ and $q = \Theta(d k_q)$. Hence, $q = \Theta(d \log_q n)$ or $q \log q = \Theta(d \log n)$. Since $q \geq d$, we get that $q = O(d \log_d n)$. Note that $t = q t_q = q^2$, therefore $t = O(d^2 \log_d^2 n)$. The corresponding binary code has constant column weights $w = t_q = q$. Note that
$q \log q = \Theta(d \log n)$ is related to the famous Lambert $W$ function \cite{lambert} and using $W(x) \geq \log x - \log \log x$, we get $q = \Omega\left(  \frac{d \log n}{\log(d \log n)} \right)$ or equivalently $w = \Omega\left( \frac{d \log n}{\log d + \log \log n}\right) $.
From our earlier discussion on MDS codes, recall that achieving a minimum distance of $r_q + 1$ requires that for any arbitrary $k_q$ rows of this code, the chosen $k_q \times q^{k_q}$ matrix must include all $q^{k_q}$ possible assignments of $q$-nary symbols in the columns. It follows that any row of a Reed-Solomon code must include every $q$-nary symbol an equal number of times. More precisely, each $q$-nary symbol must present
$q^{k_q - 1}$ times in all rows. Therefore,
the corresponding binary code has a constant row weight of $\rho = n/q$. Since $q = O(d \log_d n)$, it follows that $\rho = \Omega\left( \frac{n}{d \log_d n}\right) $.
\end{proof}

A different line of work introduced by Porat and Rothscheld in \cite{porat2008} constructs $t \times n$ $d$-{disjunct} matrices with $t = O( d^2 \log n)$.
Their approach is based on $q$-nary codes that meet the Gilbert-Varshamov bound
where the alphabet size is $q = \Theta(d)$. As in the Kautz-Singleton construction, their inner code is the identity code. The resulting binary code has the property that all the columns have the same weight of $w = \Theta( d \log n)$. Furthermore, the maximum row weight satisfies $\rho_{\max} = \Omega(n / d)$.

\begin{theorem}
The explicit construction by Porat and Rothscheld in \cite{porat2008} achieves a $t \times n$ $d$-{disjunct} matrix where $t = O( d^2 \log n)$ with constant column weight $w = \Theta( d \log n)$ and maximum row weight $\rho_{\max} = \Omega(n / d)$.
\end{theorem}

Compared to the Kautz-Singleton construction, one can observe that the Porat and Rothscheld's construction is better in the regime where $d = O(\textnormal{poly}(\log n))$. However, the Kautz-Singleton construction is better when $d = \Theta(n^\alpha)$ for some constant $\alpha \in (0, 1/2)$. In this regime, the Kautz-Singleton construction meets the fundamental lower bound and is therefore order-optimal. 

\subsection{Noisy Test Outcomes}

We have so far discussed the setting in which the test outcomes are always correct, i.e., there is no noise in the measurement process. However, in many practical applications such as drug discovery and DNA library screening, testing errors are present \cite{KAINKARYAM, du2000combinatorial}. Naturally, the aforementioned definitions and techniques can be extended so that one can identify the defective items even with certain number of faulty test outcomes. The following definition extends the notion of disjunctiveness in the presence of error in the measurement process \cite{du2000combinatorial}.
\begin{definition}
A $t \times n$ binary matrix $M$ is called $(d, \nu)$-disjunct ($\nu$-error detecting $d$-disjunct) if $|\textnormal{supp}(M_i) \backslash \cup_{j \in S} \textnormal{supp}(M_j)  | > \nu$ for every set $S$ of columns with $|S| \leq d$ and every $i \in [n]\backslash S$.
\end{definition}

We note that $(d, 0)$-disjunct is simply $d$-disjunct and a $(d, \nu)$-disjunct matrix can detect up to $\nu$ errors and can correct up to $\lfloor \nu/2 \rfloor$ errors in the test measurements. The latter can be done by simply modifying the cover decoder to incorporate the noise as follows.
\begin{proposition}
Let the cover decoder scan through the columns of ${M}$ and eliminate all items belonging to at least 
$\lceil \nu/2 \rceil + 1$ negative tests and return the remaining items. The cover decoder correctly identifies all defective items without any error if $M$ is $(d, \nu)$-disjunct in the case when the test outcomes have up to $\lfloor \nu/2 \rfloor$ errors.
\end{proposition}

We similarly define $t(d, \nu, n)$ to be the smallest $t$ needed for a binary $t \times n$ matrix $M$ to be $(d, \nu)$-disjunct.

\section{Main Results}
\label{sec:mainres}

In this section, we formally present our results for both the sparse codewords and sparse tests settings. We begin with the sparse codewords setting.

\subsection{Sparse Codewords}
\label{sec:sparse_codewords}

In the sparse codewords setting, we focus on a model where each item can participate in a limited number of tests. This is equivalent to restricting the codewords (columns of $M$) to have a limited number of ``1''s. 
For ease of presentation, we begin with the noiseless case and then extend our results to the more general noisy case in what follows. 

Recall, from the discussion in the preliminaries section, that if the codewords have a Hamming weight that is bounded by $d$, one cannot do better than the identity matrix; i.e., $t = n$. Hence, we are interested in the regime where $w_{\max} > d$.

Given that it is impossible to achieve $t < n$ when $w_{\max} \leq d$, it is natural to ask: what happens when $w_{\max} = d+1$? We recall from the preliminaries section that the Kautz-Singleton construction provides a constant column weight ($w = t_q$) group testing code that is $\left\lfloor\frac{t_q-1}{k_q -1} \right \rfloor$-{disjunct}. By choosing $k_q = 2$ and $t_q = d + 1$ we get a $d$-disjunct matrix
with $t = (d+1) \sqrt{n}$ tests and $w = d+1$ column weights when $q \geq t_q$ is satisfied. Therefore the natural question is how good this construction is in terms of the required number of tests for a $d$-disjunct matrix with $w_{\max} \leq d + 1$. The following theorem presents our first result answering this question.
\begin{theorem}\label{thm:sparsecodewords_special}
For all integers $d, n \geq 2$ such that $d + 1 \leq \sqrt{n}$ and $\sqrt{n}$ is a prime power,  the Kautz-Singleton construction provides a $t \times n$ matrix that is $d$-disjunct with constant column weight $w = d + 1$ and
\begin{align*}
t = (d + 1) \sqrt{n}.
\end{align*}

On the other hand, for all integers $d, n \geq 2$, a $t \times n$ matrix that is $d$-disjunct with maximum column weight
$w_{\max} \leq d+1$ must satisfy
\begin{align*}
t \geq \min \left\lbrace \sqrt{nd(d+1)} , n \right\rbrace.
\end{align*}
\end{theorem}
\begin{proof}
We begin with the lower bound. Let ${M}$ be a $t \times n$ $d$-{disjunct} matrix with $w_{\max} \leq d+1$. We will separate the columns of $M$ into two groups $\mathcal{N}_1, \mathcal{N}_2 \subseteq [n]$ such that $\mathcal{N}_1 \cup \mathcal{N}_2 = [n]$ and $\mathcal{N}_1 \cap \mathcal{N}_2 = \emptyset$. We define a row $i \in [t]$ to be private for a column $j$, if $j$ is
the only column in the matrix having a one on row $i$. If a column ${M}_j$ has weight at most $d$, then it must have at least one private row, otherwise we can find at most $d$ columns such that their union will span ${M}_j$ which contradicts with $d$-disjunctiveness. Now consider all columns that have weight equal to $d+1$. It is possible that some of them also have private rows. Hence, we construct the first set $\mathcal{N}_1$ such that it includes the columns whose weight is less than or equal to $d$ and the ones that have weight equal to $d+1$ such that they have at least one private row. The second set $\mathcal{N}_2$ consists of the rest of the columns; i.e., the ones that have weight equal to $d+1$ and do not have any private row. Defining $w_j$ to be weight of the column $j$ for $1 \leq j \leq n$, more formally we define
\begin{align*}
&\mathcal{N}_1 \coloneqq \left\lbrace j \in [n] \ | \ w_j \leq d \ \textnormal{or} \ w_j = d+1 \ \textnormal{and ${M}_j$ has at least one private row}
\right\rbrace, \\
&\mathcal{N}_2 \coloneqq \left\lbrace j \in [n] \ | \ w_j = d+1 \ \textnormal{and ${M}_j$ has no private row}
\right\rbrace.
\end{align*}
Note that by construction, $\mathcal{N}_1 \cup \mathcal{N}_2 = [n]$ and $\mathcal{N}_1 \cap \mathcal{N}_2 = \emptyset$, hence $n = |\mathcal{N}_1| + |\mathcal{N}_2|$. In the following, we will bound the size of both sets $\mathcal{N}_1$ and $\mathcal{N}_2$.

We note that each column in the set $\mathcal{N}_1$ has at least one private row and by definition of the private row it cannot be shared by two distinct columns. Since there could be at most $t$ private rows, we have $|\mathcal{N}_1| \leq t$.

We now consider the set $\mathcal{N}_2$. We generalize the definition of the private row to the private set as follows. A private set for a column is defined as a set of position of ones such that no other column can cover these positions by itself; i.e., no other column has all ones in these positions. We claim that all size-2 sets of positions of ones of a column in set $\mathcal{N}_2$ must be private; i.e., all pairs of positions of ones must be private for a column in set $\mathcal{N}_2$. We prove this by contradiction. Assume there exists a column in the set $\mathcal{N}_2$ such that it has at least one pair of positions of ones which is not private. This means that there exists another column which can cover these positions. Note that any column in the set $\mathcal{N}_2$ has weight $d+1$ and has no private row, therefore, there are $d-1$ positions of ones except this pair and we can find at most $d-1$ columns that can cover all these positions. It follows that we can find at most $d$ columns that can cover all $d+1$ positions of ones of this column which contradicts with the $d$-disjunctiveness. Note that there are $\binom{d+1}{2}$ number of pairs of position of ones for any column in $\mathcal{N}_2$ and by definition of private set it cannot be shared by two distinct columns. We further note that each column in the set $\mathcal{N}_1$ will have a private row and it must be the case that the columns in the set $\mathcal{N}_2$ must have a zero in these rows, therefore, there could be at most $\binom{t - |\mathcal{N}_1|}{2}$ number of private pairs. Hence, we have
\begin{align*}
|\mathcal{N}_2| \binom{d+1}{2} \leq \binom{t - |\mathcal{N}_1|}{2}.
\end{align*}
Therefore,
\begin{align*}
|\mathcal{N}_2| d (d+1) \leq (t - |\mathcal{N}_1|)(t - |\mathcal{N}_1| - 1) \leq (t - |\mathcal{N}_1|)^2.
\end{align*}
Defining $n_1 \triangleq |\mathcal{N}_1|$, this gives us
\begin{align}
t \geq n_1 + \sqrt{(n-n_1)d(d+1)}
\label{concave}
\end{align}
Note that $0 \leq n_1 \leq t \leq n$. One can take the second derivative of the right-hand side of \eqref{concave} and observe that it is negative for $0 \leq n_1 \leq t \leq n$ which means it is a concave function of $n_1$ and it will be minimum at either $n_1 = 0$ or $n_1 = t$. Therefore,
\begin{align*}
t \geq \min  \left\lbrace \sqrt{nd(d+1)} , t +  \sqrt{(n-t)d(d+1)} \right\rbrace.
\end{align*}
Noting that $t \geq t +  \sqrt{(n-t)d(d+1)}$ only when $t = n$, one can observe that
\begin{align*}
t \geq \min \left\lbrace \sqrt{nd(d+1)} , n \right\rbrace.
\end{align*}


For the achievability, we use the Kautz-Singleton construction. We choose $w = t_q = d+1$ and $k_q = 2$. Since $n = q^{k_q}$, we obtain $q = \sqrt{n}$ and therefore $t = (d+1)\sqrt{n}$. Note that $\lambda_{\max} 
= k_q - 1 = 1$, hence $w_{\min} \geq d \lambda_{\max} + 1$ is satisfied which is sufficient to achieve a $d$-disjunct matrix. In order to satisfy the requirement $q \geq t_q$ where $q$ is any prime power, we must ensure that $d+1 \leq \sqrt{n}$ and $q = \sqrt{n}$ must be a prime power. This completes the proof for the achievability.
\end{proof}

A few comments are in order. First, Theorem \ref{thm:sparsecodewords_special} shows that by increasing $w_{\max}$ from $d$ to $d +1$, we suddenly get $t = \Theta(d \sqrt{n})$ instead of $t = n$. Second, the achievability result in Theorem \ref{thm:sparsecodewords_special} is obtained by changing the field size from $q = O(d \log_d n)$ to $q = \sqrt{n}$ in the Kautz-Singleton construction. The Kautz-Singleton construction is strictly suboptimal in the classical case when $d = O(\textnormal{poly}(\log n))$. It is interesting that a simple modification of this well known construction makes it optimal in this case (even up to an almost matching constant).

We next investigate the more general case where the codeword weights are bounded by $w_{\max} \leq l d + 1$ for some integer $l > 1$. We note that by choosing $k_q = l + 1$ and $t_q = l d + 1$ we get a $d$-disjunct matrix with $t = (ld+1) \sqrt[l+1]{n}$ tests and $w = l d + 1$ column weights using the Kautz-Singleton construction when $q \geq t_q$ is satisfied. In this case we can show that this construction is nearly optimal.
\begin{theorem} \label{thm:sparsecodewords}
For all integers $d, n, l \geq 2$ such that $l d  + 1 \leq \sqrt[l+1]{n}$ and $\sqrt[l+1]{n}$ is a prime power, the Kautz-Singleton construction provides a $t \times n$ matrix that is $d$-disjunct with constant column weights $w = l d + 1$ and
\begin{align*}
t = (l d + 1) \sqrt[l+1]{n},
\end{align*}

On the other hand, for all integers $d, n, l \geq 2$, a $t \times n$ matrix that is $d$-disjunct with maximum column weight
$w_{\max} \leq l d + 1$ must satisfy
\begin{align*}
t \geq \left(  \dfrac{(l-1)^{l+1}(d-1)^{l+1}}{ 2 e^l (l-1) (d-1)^{l-1} + 1}  \right)^{1/(l+1)} \sqrt[l+1]{n}.
\end{align*}
\end{theorem}
\begin{proof}
We begin with the lower bound. Let ${M}$ be a $t \times n$ $d$-{disjunct} matrix with $w_{\max} \leq ld+1$. We similarly separate the columns of $M$ into $l+1$ groups and construct the sets $\mathcal{N}_i$ for $i = 1, \ldots, l+1$ such that $\mathcal{N}_1 \cup \ldots \cup \mathcal{N}_{l+1} = [n]$ and $\mathcal{N}_i \cap  \mathcal{N}_j = \emptyset$ for any $i, j \in [l+1]$ such that $i \neq j$. We construct the sets $\mathcal{N}_1, \ldots, \mathcal{N}_{l+1}$ as follows.
We keep the first set $\mathcal{N}_1$ as the columns whose weight is less than or equal to $d$ and the ones that have weight equal to $d+1$ such that they have at least one private row. For $i = 2, \ldots, l$, the set $\mathcal{N}_i$ consists of the columns that satisfies one of the following two conditions: Either its weight is between $(i-2)d + 2 $ and $(i-1)d + 1$ and it has no private set of size $i-1$ or between $(i-1)d + 2$ and $i d + 1$ and it has at least one private set of size $i$. Finally, the last set $\mathcal{N}_{l+1}$ consists of the columns whose weight is between $(l-1)d + 2 $ and $l d + 1$ and they have no private set of size $l$. More formally,
\begin{align*}
&\mathcal{N}_1 \coloneqq \left\lbrace j \in [n] \ | \ w_j \leq d \ \textnormal{or} \ w_j = d+1 \ \textnormal{and ${M}_j$ has at least one private row}
\right\rbrace, \\
&\mathcal{N}_i \coloneqq \lbrace j \in [n] \ | \  (i-2)d + 2\leq w_j \leq (i-1)d + 1 \ \textnormal{and ${M}_j$ has no private set of size $i-1$} \
\\ &  \ \ \ \ \ \ \ \ \ \ \ \ \ \ \
\textnormal{or}
\\ &  \ \ \ \ \ \ \ \ \ \ \ \ \ \ \
(i-1)d + 2 \leq w_j \leq i d + 1 \ \textnormal{and ${M}_j$ has at least one private set of size $i$} \rbrace,
\\ & \textnormal{for $i = 2, \ldots, l$},
\\ & \mathcal{N}_{l+1} \coloneqq \lbrace j \in [n] \ | (l-1)d + 2\leq w_j \leq l d + 1 \ \textnormal{and ${M}_j$ has no private set of size $l$} \rbrace.
\end{align*}
Note that by construction, $\mathcal{N}_1 \cup \ldots \cup \mathcal{N}_{l+1} = [n]$ and $\mathcal{N}_i \cap  \mathcal{N}_j = \emptyset$ for any $i, j \in [l+1]$ such that $i \neq j$, hence $n = |\mathcal{N}_1| + \ldots + |\mathcal{N}_{l+1}|$. In the following, we will bound the size of these sets.

Recalling the discussion in the previous section, we have $|\mathcal{N}_1| \leq t$. Consider the sets $\mathcal{N}_i$ for $i = 2, \ldots, l$. For any column $j \in \mathcal{N}_i$, if we have $(i-1)d + 2 \leq w_j \leq i d + 1$, then by construction ${M}_j$ has at least one private set of size $i$. For the case
$(i-2)d + 2\leq w_j \leq (i-1)d + 1$, we claim that all the sets of positions of ones of size $i$ must be private for the column ${M}_j$. We similarly show this by contradiction. Assume there exists a set of positions of ones of size $i$ such that it is not private for the column ${M}_j$. Then we can find a column that can cover these positions. Since by construction of set ${N}_i$, the column ${M}_j$ has no private set of size $i-1$, one can find at most $ ((i-1)d + 1 - i)/(i-1) = d - 1 $ columns that will cover the rest of the positions of ones. Hence we have at most $d$ columns covering the column ${M}_j$ which contradicts the $d$-disjunctiveness. Therefore, we obtain that all the columns in the set ${N}_i$ must have at least one private set of size $i$. Since the private sets cannot be shared among the columns and we have at most $\binom{t}{i}$ private sets of size $i$, it yields $|\mathcal{N}_i| \leq \binom{t}{i}$.  For the last set $\mathcal{N}_{l+1}$, similar arguments apply and for each column, it should be the case that all the  set of positions of ones of size $l+1$ must be private. Since $w_j \geq (l-1)d + 2$ for $j \in \mathcal{N}_{l+1}$, we have $|\mathcal{N}_{l+1}| \binom{(l-1)d + 2}{l+1} \leq \binom{t}{l+1}$. Therefore,
\begin{align*}
n & = |\mathcal{N}_1| + \ldots + |\mathcal{N}_{l+1}| \\ & \leq \sum \limits_{i=1}^{l} \binom{t}{i}  + \dfrac{\binom{t}{l+1}}{\binom{(l-1)d + 2}{l+1}}
\\ & \overset{(i)}{\leq}
\left( \dfrac{e t}{l} \right)^l + \dfrac{t \ldots (t - l)}{ ((l-1)d + 2) \ldots ((l-1)(d-1) + 1)}
\\ & \overset{(ii)}{\leq}
\dfrac{e^l t^l}{l^l}  + \dfrac{t^{l+1}}{ ((l-1)(d-1))^{l+1} }
\\ & \overset{(iii)}{\leq}
\dfrac{e^l t^{l}}{(l-1)^l} \dfrac{t}{(d-1)^2/2}  + \dfrac{t^{l+1}}{ ((l-1)(d-1))^{l+1} }
\\ & =
t^{l+1} \left(  \dfrac{2 e^l}{(l-1)^l (d-1)^2} + \dfrac{1}{(l-1)^{l+1}(d-1)^{l+1}}  \right)
\\ & =
t^{l+1} \left(  \dfrac{ 2 e^l (l-1) (d-1)^{l-1} + 1}{(l-1)^{l+1}(d-1)^{l+1}}  \right)
\end{align*}
where $(i)$ is due to the inequality $\sum \limits_{i=0}^{l} \binom{t}{i} \leq \left( \frac{e t}{l} \right)^l$
for $t \geq l \geq 1$, $(ii)$ is bounding all the terms in the numerator by $t$ and denominator by $(l-1)(d-1)$ and in $(iii)$  we use \eqref{simpleLB} and $\binom{d+2}{2} \geq \frac{(d-1)^2}{2}$. This completes the lower bound.

For the achievability, we use the Kautz-Singleton construction. We choose $w = t_q = ld+1$ and $k_q = l+1$. Since $n = q^{k_q}$, we obtain $q = \sqrt[l+1]{n}$ and therefore $t = (ld+1)\sqrt[l+1]{n}$. Note that $\lambda_{\max} 
= k_q - 1 = l$, hence $w_{\min} \geq d \lambda_{\max} + 1$ is satisfied which is sufficient to achieve a $d$-disjunct matrix. In order to satisfy the requirement $q \geq t_q$ where $q$ is any prime power, we must ensure that $ld+1 \leq \sqrt[l+1]{n}$ and $q = \sqrt[l+1]{n}$ must be a prime power. This completes the proof for the achievability.
\end{proof}

Note that as we increase the weights as a multiple of $d$ (i.e., $w_{\max} = ld +1$), the minimum number of required tests decreases exponentially in $l$. As we see from Theorem \ref{thm:sparsecodewords}, for a fixed $l$ the lower bound we get is $\Theta\left(d^{\frac{2}{l+1}} \sqrt[l+1]{n}\right)$ whereas the upper bound is $\Theta(d \sqrt[l+1]{n})$. While we have a matching lower and upper bounds in terms of the scaling with respect to $n$, there is an increasing gap of $d^{\frac{l-1}{l+1}}$ between them, which approaches $d$ for large $l$.

We continue our discussion with the noisy case. As we have seen that it is impossible to achieve $t < n$ when $w_{\max} \leq d$ in the noiseless case, a similar result can be observed in the noisy case as well. 
Our next result extends this to the noisy setting with an arbitrary noise parameter $v$.
\begin{proposition}
If $w_{\max} \leq d + \nu$, then $t(d, \nu, n) = (\nu+1) n$.
\label{prp:sc_noisy_trivial} 
\end{proposition}
The proof of the above theorem can be found in Appendix \ref{proof:sc_noisy_trivial}. Proposition \ref{prp:sc_noisy_trivial} similarly shows that one cannot do better than individual testing corresponding to the more general noisy setting if the codeword weights are bounded by $d + \nu$. 

We note that 
it is sufficient to have $w_{\min} \geq d \lambda_{\max} + \nu + 1$ to obtain a $(d, \nu)$-disjunct matrix. 
We can employ the Kautz-Singleton construction and fix $k_q = 2$ and $t_q = d+\nu+1$ to get a $(d, \nu)$-disjunct matrix with $t = (d + \nu + 1)\sqrt{n}$ tests and $w = d + \nu + 1$ column weights when $q \geq t_q$ is satisfied. The following theorem shows that this is order-optimal when $w_{\max} \leq d+\nu+1$.
\begin{theorem}\label{thm:sc_noisy_first}
For all integers $d, n \geq 2$ and $\nu \geq 0$, a $t \times n$ matrix that is $(d, \nu)$-disjunct with maximum column weight
$w_{\max} \leq d+\nu+1$ must satisfy
\begin{align*}
t \geq \min \{(\nu+1) n, \sqrt{(d+\nu)(d+\nu+1)n}\}.
\end{align*}
\end{theorem}
The proof of the above theorem can be found in Appendix \ref{proof:sc_noisy_first}. It is interesting to observe that by increasing $w_{\max}$ from $d + \nu$ to $d + \nu + 1$, we are able to reduce to $t = \Theta((d+\nu) \sqrt{n})$ from $t = (\nu + 1)n$. Going further, we can generalize this to the case where the codeword weights are bounded by $w_{\max} \leq l d + \nu + 1$ for some integer $l > 1$. Fixing $k_q = l+1$ and $t_q = ld+\nu+1$ in the Kautz-Singleton construction provides us with a $(d, \nu)$-disjunct matrix that has 
$t = (ld+\nu+1) \sqrt[l+1]{n}$ tests and $w = l d + \nu + 1$ column weights. The next theorem shows that 
this construction is nearly optimal.
\begin{theorem}\label{thm:sc_noisy_general}
For all integers $d, n, l \geq 2$ and $\nu \geq 0$, a $t \times n$ matrix that is $(d, \nu)$-disjunct with maximum column weight
$w_{\max} \leq ld+\nu+1$ must satisfy
\begin{align*}
t \geq \left( \dfrac{2 e^l}{(d+v)^2 (l-1)^l} + \dfrac{1}{((l-1)(d-1) + \nu)^{l+1}} \right)^{-1/(l+1)} \sqrt[l+1]{n}.
\end{align*}
\end{theorem}
The proof of the above theorem can be found in Appendix \ref{proof:sc_noisy_general}. Similar to the noiseless case, we have a matching lower and upper bounds in terms of the scaling with respect to $n$ and order-wise the lower bound in Theorem \ref{thm:sc_noisy_general} is $\Theta\left((d+\nu)^{\frac{2}{l+1}} \sqrt[l+1]{n}\right)$ whereas the Kautz-Singleton construction provides $\Theta((d+\nu) \sqrt[l+1]{n})$ tests. 

\subsection{Sparse Tests}
\label{sec:sparse_tests}

In the sparse tests setting, we focus on a model where each test can include a limited number of items. In other words, we restrict the row weights of $M$, and derive lower and upper bounds on the minimum number of tests so that $M$ is a $(d, \nu)$-disjunct matrix in the more general noisy case (including the noiseless setting as a special case under $\nu = 0$).

Our first theorem provides a fundamental lower bound on the minimum required number of tests under a row weight constraint and an upper bound which is again based on the Kautz-Singleton construction.
\begin{theorem}\label{rowLB}
For all integers $d, n \geq 2$ and $\nu, l \geq 0$ such that $l d  + \nu + 1 \leq \sqrt[l+1]{n}$ and $\sqrt[l+1]{n}$ is a prime power,  the Kautz-Singleton construction provides a $t \times n$ matrix that is $(d, \nu)$-disjunct with constant row weights $\rho = n^{\frac{l}{l+1}}$ and
\begin{align*}
t = (l d + \nu + 1) \sqrt[l+1]{n}.
\end{align*}

On the other hand, for all integers $d, n \geq 2$ and $\nu \geq 0$, a $t \times n$ matrix that is $(d, \nu)$-disjunct with maximum row weight
$\rho_{\max}$ must satisfy
\begin{align*}
t \geq
\left\{ \begin{array}{ll}
\dfrac{(d+\nu+1) n}{\rho_{\max}} & \mbox{if $\rho_{\max} > \dfrac{d+\nu+1}{\nu+1}$,}\\
(\nu + 1) n & \mbox{if $\rho_{\max} \leq \dfrac{d+\nu+1}{\nu+1}$}\end{array} \right.
\end{align*}
\end{theorem}
\begin{proof}
We begin with the lower bound.
Suppose we have a $t \times n$ matrix that is $(d, \nu)$-disjunct and all the rows are bounded with weight $\rho_{\max}$. We consider the columns that have weight less than or equal to $d+\nu$. From the preceding discussions, all these columns must have at least $\nu + 1$ private rows. Let us delete these columns and $\nu + 1$ private rows for each column. Note that if a column has more than $\nu + 1$ private rows, we can arbitrarily choose any $\nu + 1$ of them and delete it. We will be deleting $\nu + 1$  private rows per such column and it is possible since they all have at least $\nu + 1$  private rows and a private row cannot be shared by two distinct columns.

Let us denote $t_1$ by the number of columns whose weight is less than or equal to $d + \nu$. From the discussion of private rows, it follows that $0 \leq t_1 \leq t/(\nu+1)$.
After the deletion operation, the dimension of the resulting matrix is $ (t - t_1 (\nu+1)) \times (n - t_1) $ and it is still $(d, \nu)$-disjunct since we are only deleting the zero-entries of the rest of the columns, therefore, the resulting matrix must still satisfy the $(d, \nu)$-disjunctiveness. We also note that the rows are still bounded with weight $\rho_{\max}$.

We observe that in the resulting matrix, all the columns will have at least weight of $d + \nu + 1$ and therefore the total number of ones in the resulting matrix can be lower bounded by $(d + \nu + 1)(n - t_1)$ and upper bounded by $\rho_{\max} (t - t_1 (\nu + 1))$. Hence,
\begin{align}
& \rho_{\max} (t - t_1 (\nu+1)) \geq (d + \nu + 1)(n - t_1), \nonumber \\
& \rho_{\max} t \geq (d + \nu + 1) n  + t_1 \left(\rho_{\max}(\nu + 1) - (d + \nu + 1) \right).
\label{rowmainineq}
\end{align}
If $\rho_{\max} \leq \frac{d+\nu+1}{\nu+1}$, then from $t_1 \leq t/(\nu + 1)$ and \eqref{rowmainineq} we have
\begin{align*}
\rho_{\max} t \geq (d + \nu + 1) n  + \frac{t}{\nu+1} \left(\rho_{\max}(\nu + 1) - (d + \nu + 1) \right).
\end{align*}
It follows that $t \geq (\nu + 1)n$. On the other hand, if $\rho_{\max} > \frac{d+\nu+1}{\nu+1}$, then from $t_1 \geq 0$ and \eqref{rowmainineq} we have
\begin{align*}
\rho_{\max} t \geq (d + \nu + 1) n.
\end{align*}
This yields that $t \geq \frac{(d + \nu + 1) n}{\rho_{\max}}$.

For the achievability, we use the Kautz-Singleton construction. Recall from the preliminaries section that the Kautz-Singleton construction provides a binary code with constant row weight $\rho = n/q$. We choose $t_q = ld+\nu+1$ and $k_q = l+1$. Since $n = q^{k_q}$, we obtain $q = \sqrt[l+1]{n}$ hence $t = (ld+\nu+1)\sqrt[l+1]{n}$ and $\rho = n^{\frac{l}{l+1}}$. In order to satisfy the requirement $q \geq t_q$ where $q$ is any prime power, we must ensure that $ld+\nu+1 \leq \sqrt[l+1]{n}$ and $q = \sqrt[l+1]{n}$ must be a  prime power. This completes the proof for the achievability.
\end{proof}
Observe that for any fixed integer $l \geq 1$ that satisfies the conditions stated in Theorem \ref{rowLB}, the number of tests we get using the Kautz-Singleton construction is $\Theta((d + \nu) \sqrt[l+1]{n})$ with constant row weight $\rho = n^{\frac{l}{l+1}}$. Plugging this in the lower bound of Theorem \ref{rowLB}, the required number of tests is also $\Theta((d + \nu) \sqrt[l+1]{n})$. It is interesting to note that the Kautz-Singleton construction is order-optimal in this setting. It has been known in the group testing literature that if the weights of the columns are bounded by $d$, one cannot do better than the identity matrix; i.e., $t = n$. Theorem \ref{rowLB} states an analogous result for the case with row weight constraint: if the weights of the rows are bounded by $\frac{d + \nu + 1}{\nu+1}$, we have $t  = (\nu + 1)n$ which means that we cannot do better than the individual testing. Another very interesting result that can be obtained from Theorem \ref{rowLB} is that for the special case where $l = 1$, the Kautz-Singleton construction is optimal to the exact constants.
\begin{corollary}
For all integers $d, n \geq 2$ and $\nu \geq 0$, the Kautz-Singleton construction provides an optimal (to the exact constants) $(d, \nu)$-disjunct matrix under the maximum row weight constraint $\rho_{\max} \leq \sqrt{n}$.
\end{corollary}

We emphasize that the Kautz-Singleton construction in Theorem \ref{rowLB} provides us with codes that have constant row weight of $ n^{\frac{l}{l+1}}$; i.e., when $\rho$ is a fractional power of $n$ in the form $\frac{l}{l+1}$ in the interval $[1/2, 1)$. It is natural to ask whether there exist group testing codes with  $\rho = n^{\alpha}$ for an arbitrary $\alpha \in (0, 1)$ that achieves the lower bound in Theorem \ref{rowLB}. The following theorem shows the existence of such codes when $d = O(\textnormal{poly}(\log n))$ by using a random construction.
\begin{theorem}\label{random}
In the regime where $d = O(\textnormal{poly}(\log n))$, there exists a $t \times n$ matrix that is $(d, \nu)$-disjunct with a maximum row weight $\rho_{\max} = \Theta(n^{\alpha})$, for any $\alpha \in (0, 1)$, such that
\begin{align*}
t = O\left( (d + \nu) n^{1-\alpha}\right).
\end{align*}
\end{theorem}
\begin{proof}
Let us construct the matrix ${M}$ randomly as follows. For a fixed $\alpha \in (0, 1)$, we take $t = c (d + \nu) n^{1 - \alpha} $ for some constant $c > 0$ that we will fix later. We set the size of the matrix ${M}$ as $t \times n$ and choose the columns of this matrix uniformly at random among the codewords of size $t$ with weight $w$ where we set $w = c (d + \nu)$.

We next calculate the probability of not having a $(d, \nu)$-disjunct matrix. Let us fix $d + 1$ columns of the matrix ${M}$
and denote them as ${M}_1, \ldots, {M}_{d+1}$. Let us further fix a single column among them, say ${M}_{d+1}$. The probability of violating the condition $\left\vert\textnormal{supp}(M_{d+1}) \backslash \cup_{j \in S} \textnormal{supp}(M_j)  \right\vert > \nu$ where $S= \{1,2 \ldots, d\}$ can be bounded as 
\begin{align*}
\mathbb{P} \left( \left \vert \textnormal{supp}(M_{d+1}) \backslash \cup_{j \in S} \textnormal{supp}(M_j)  \right \vert \leq \nu \right)
\leq \dfrac{\binom{dw}{w-\nu} \binom{t-w+\nu}{\nu}}{\binom{t}{w}},
\end{align*}
since ${M}_1, \ldots, {M}_{d}$ can have at most $d w$ non-intersecting number of ones and the numerator is an upper bound on the total number of codewords that will have at least $w-\nu$ intersection between $\textnormal{supp}(M_{d+1})$ and $\cup_{j \in S} \textnormal{supp}(M_j) $. Using union bound, the probability that the matrix ${M}$ does not satisfy the $(d, \nu)$-disjunctiveness property can be bounded as
\begin{align*}
\mathbb{P} \left( \textnormal{${M}$ is not $(d, \nu)$-disjunct} \right) \leq (d+1)
\binom{n}{d+1} \dfrac{\binom{dw}{w-\nu} \binom{t-w+\nu}{\nu}}{\binom{t}{w}}.
\end{align*}
We can further bound this as
\begin{align}
\mathbb{P}  \left( \textnormal{${M}$ is not $d$-disjunct} \right) & \overset{(i)}{\leq}
(d+1) \left( \dfrac{n e}{d+1} \right)^{d+1}
\dfrac{\left( \dfrac{d w e}{w-\nu} \right)^{w-\nu} \left( \dfrac{(t-w+\nu)e}{\nu} \right)^{\nu}}{\left( \dfrac{t}{w} \right)^w} \nonumber \\
& = (d+1) \left( \dfrac{n e}{d+1} \right)^{d+1}  \dfrac{ \left( \dfrac{d w}{w-\nu} \right)^{w-\nu} \left( \dfrac{t-w+\nu}{\nu} \right)^{\nu} e^w}{n^{w(1-\alpha)}} \nonumber \\
& \overset{(ii)}{\leq}  \dfrac{n^{d+1} e^{d+w+1}}{d^d} \cdot \dfrac{(2d)^{w-\nu} (t/\nu)^{\nu}}{n^{w(1-\alpha)}} 
\label{ran}
\end{align}
where $(i)$ is due to the inequality
$\left( \frac{n}{k} \right)^k \leq \binom{n}{k} \leq  \left( \frac{n e}{k} \right)^k$ and $(ii)$ is due to $w-\nu \geq \nu/2$ for $c\geq 2$ and $t-w+\nu \leq t$ and $d+1 \geq d$. Taking the logarithm of the last term in \eqref{ran} gives us
\begin{align*}
(d+1) \log n + d + w + 1 + (w - \nu) \log 2 + (w - \nu - d) \log d + \nu \log(t/\nu) - w(1 - \alpha)\log n
\leq - \tilde{c} (d+v) \log n
\end{align*}
that holds with a constant $\tilde{c}>0$ for sufficiently large $n$ and appropriately chosen constant $c>4/(1-\alpha)$ when $d = {O}(\textnormal{poly}(\log n))$. Hence,
\begin{align}
\mathbb{P}  \left( \textnormal{${M}$ is not $(d, \nu)$-disjunct} \right) \leq
n^{-\tilde{c} (d+v)}.
\label{probdisjunct}
\end{align}
We next investigate the weights of the rows of the matrix ${M}$. We consider the first row. Note that by our random construction, it follows that each
entry in the first row is independent and identically distributed with Bernoulli distribution where the probability of having one is $\frac{\binom{t-1}{w-1}}{\binom{t}{w}} = \frac{w}{t}$. Denoting $\rho_1$ as the weight of the first row, we have $\mathbb{E}[\rho_1] = \frac{w}{t} n = n^{\alpha}$. Using Hoeffding's inequality along with union bound, we achieve the following upper bound on the probability that there exists a row with its weight deviating from $n^{\alpha}$
\begin{align}
\mathbb{P}(\textnormal{$\exists$ $i \in [t]$ s.t.  $|\rho_{i} - n^{\alpha}| \geq \delta  n^{\alpha}$}) & \leq t 2 e^{-2 n^{\alpha} \delta^2} \nonumber \\ & = 2 c (d+\nu) n^{1 - \alpha} e^{- 2 n^{\alpha} \delta^2}
\label{deviate}
\end{align}
for some fixed constant $0 < \delta < 1$. For sufficiently large $n$, the right-hand side of \eqref{deviate} can be bounded as
$e^{- \bar{c} n^{\alpha}}$ for some constant $\bar{c} > 0$.

Together with \eqref{probdisjunct}, using union bound, it follows that with probability approaching to 1, we get a $(d, \nu)$-disjunct matrix with row weight $\rho = \Theta(n^\alpha)$ and $t = \Theta((d +\nu) n^{1- \alpha})$ tests.
\end{proof}
In the regime where $d = O(\textnormal{poly}(\log n))$, the lower bound in Theorem \ref{rowLB} suggests that the minimum number of tests is $\Omega( (d+\nu) n^{1- \alpha})$ when $\rho = \Theta( n^{\alpha} )$ for some $\alpha \in (0, 1)$. The randomized construction in Theorem \ref{random} proves that there exist codes that achieve $t = \Theta( (d+\nu) n^{1- \alpha})$. This matches the lower bound in Theorem \ref{rowLB}.

\section{Encoding \& Decoding}
\label{sec:decoding}

We have so far focused on investigating the fundamental trade-offs between $t$ and $(d, \nu, n)$ under constraints on either the number of items that can participate in a test (sparse tests) or the number of tests an item can participate in (sparse codewords) without considering the encoding or decoding complexities. 
However, due to the emerging applications involving massive datasets there is a recent research effort towards low-complexity decoding schemes \cite{mahdi2009, indyk2010, ngo11, lee2016saffron, jaggi2017}.
The computational complexities of encoding and decoding might be just as critical, therefore, it is desirable not to sacrifice on encoding or decoding complexity to achieve the optimal trade-off between $t$ and $(d, \nu, n)$. In this section, we discuss the encoding and decoding complexities of the explicit constructions we presented earlier in this paper.

In the classical combinatorial group testing framework, the interest has been on designing testing strategies that can be decoded in $\textnormal{poly}(t)$-time while achieving the best known upper bound $t = O(d^2 \log n)$. Guruswami et al. present an efficiently decodable ($O(t)$ time decoding) $d$-disjunct matrix in \cite{Guruswami2004} and their constructions require $O(d^4 \log n)$ tests. The first result that achieves efficient decoding time while matching the $O(d^2 \log n)$ bound on the number of tests was recently presented in \cite{indyk2010}. Furthermore, the construction in  \cite{indyk2010} can be derandomized in the regime $d = O(\log N/\log \log N)$.  Later in \cite{ngo11} the constraint on $d$  is removed and  an explicit construction is provided that can be decoded in time $\textnormal{poly}(t)$. The main idea considered in \cite{indyk2010} was using \em list-disjunct \em matrices and a similar idea was considered in \cite{mahdi2009} to obtain explicit constructions of non-adaptive group testing schemes that handle noisy tests and return a small super-set of the defective items. 

We now show that our explicit constructions can be decoded in  $(\textnormal{poly}(d) + O(t))$-time and each entry in any codeword can be computed in space $\textnormal{poly}(\log n)$ by following a similar approach to \cite{indyk2010}. This shows that these constructions not only (nearly) achieve the fundamental lower bound in the energy constrained setting, but also do that with a favorable encoding and decoding complexity. We begin with the following result which is based on the noiseless setting.  

\begin{theorem}\label{thm:decoding}
For all integers $d, n \geq 2$, $l \geq 1$ such that $l d  + 1 \leq \sqrt[l+1]{n}$ and $\sqrt[l+1]{n}$ is a prime power, the Kautz-Singleton construction provides a $t \times n$ matrix that is $d$-disjunct with constant column weights $w = l d + 1$, constant row weights $\rho = n^{\frac{l}{l+1}}$, and $ t = (l d + 1) \sqrt[l+1]{n} $ tests.
Furthermore, the decoding can be done in time $\textnormal{poly}(d) + O(t)$ and each entry can be computed in space $\textnormal{poly}(\log n)$.
\end{theorem}

\begin{proof}
We describe the decoding procedure as follows. For an output vector $Y \in \{0, 1\}^t$, we can consider it as $Y = (Y_1, \ldots, Y_{t_q})$, a vector in $(\{0, 1\}^q)^{t_q}$ where $t_q = ld+1$ is the block length of the outer Reed-Solomon code.
Note that since we use the identity code as the inner code, for each $i \in [t_q]$, $Y_i$ will have at most $d$ ones and the position of ones will correspond to the symbols of defective items in the outer code. We now apply the following procedure. For each $i \in [t_q]$, we create the sets $S_i \subseteq [q]$ such that $S_i$ consists of the set of position of ones in $Y_i$. It follows that $\lvert S_i \rvert \leq d$, for every $i \in [t_q]$. We further have the following property. For any defective item, the corresponding codeword $(c_1, \ldots, c_{t_q})$ in the outer code must satisfy $c_i \in S_i$ for all $i \in [t_q]$ and for any non-defective item, the corresponding codeword $(c_1, \ldots, c_{t_q})$ in the outer code will include a symbol $c_i$ such that $c_i \notin S_i$. Note that this step can be done in $O(t)$ time.

The second step is to output all codewords $(c_1, \ldots, c_{t_q})$ in the outer code such that
$c_i \in S_i$ for all $i \in [t_q]$ given $S_i \subseteq [q]$ with $\lvert S_i \rvert \leq d$ for every $1 \leq i \leq t_q$. This problem is an instance of the error-free list recovery problem \cite{rudra_lr, indyk2001, zuckerman2004}. When each set $S_i$ has at most $s$ elements, it is referred to as list recovering with input lists of size $s$. It has been shown that the corresponding error-free list recovery problem can be solved in polynomial time for a $[t_q, k_q, t_q-k_q+1]_q$ Reed-Solomon code as long as the parameter $s$ satisfies $s < \lceil \frac{t_q}{k_q-1} \rceil$ \cite{rudra_lr, sudan99}. We note that in our case, we have $s = d$, $t_q = l d + 1$, and $k_q = l+1$, therefore it satisfies that
$s < \lceil \frac{t_q}{k_q-1} \rceil$. It follows that the second step can be done in time $\textnormal{poly}(t_q)$. In particular, we can use the algorithm in \cite{Alekhnovich02} that runs in time $\textnormal{poly}(d) \cdot t_q \log^2 t_q \log \log t_q$ which is $\textnormal{poly}(d)$ with our choice of $t_q$. The error-free list recovery problem that we are interested in solving is a special case of a more general problem known as soft decoding which is defined as follows. The decoder is given a set of non-negative weights corresponding to each row and each symbol ($w_{i, \alpha}$, $i \in [t_q]$, $\alpha \in [q]$) and a threshold $W \geq 0$. The decoder needs to output all codewords $(c_1, \ldots, c_{t_q})$  in $q$-ary code of block length $t_q$ that satisfy
\begin{align*}
\sum \limits_{i = 1}^{t_q} w_{i, c_i} \geq W.
\end{align*}
Note that the error-free list recovery is a special case of soft decoding under the parameters $W = t_q$ and $w_{i, \alpha} = 1$ for $\alpha \in S_i$ and $w_{i, \alpha} = 0$ otherwise. The soft decoding is related to weighted polynomial reconstruction problem which is defined as follows. For the given integer parameters $k$ and $N$ with $N$ points $(x_1, y_1), \ldots, (x_N, y_N)$ and their corresponding weights $w(x_1, y_1), \ldots, w(x_N, y_N)$, the goal is to output all polynomials of degree at most $k$ such that $\sum_{i : p(x_i) = y_i} w(x_i, y_i) \geq W$. The algorithm presented in \cite{Alekhnovich02} solves this problem and runs in time $\textnormal{poly}(d)$ translated to our case. Combining the two steps, we conclude that the decoding can be done in time $\textnormal{poly}(d) + O(t)$.

The reconstruction of the matrix with the claimed space complexity follows from the fact that any position in a Reed-Solomon codeword can be computed in space $\textnormal{poly}(k_q, \log q)$ and any bit value of the identity inner code can be computed in $O(\log q)$ space.
\end{proof} 

We can extend these results to the noisy case as well by following similar ideas and modifying the parameters accordingly. Our next result whose proof we give in Appendix \ref{proof:decoding_noisy}
shows the validity of efficient decoding in the case of arbitrary errors. 

\begin{theorem}\label{thm:decoding_noisy}
For all integers $d, n\geq 2$, $l \geq 1$, and $\nu \geq 0$ such that $l d  + (l+2) \nu + 1 \leq \sqrt[l+1]{n}$ and $\sqrt[l+1]{n}$ is a prime power, the Kautz-Singleton construction provides a $t \times n$ matrix that is $(d, \nu)$-disjunct with constant column weights $w = l d + (l+2) \nu + 1$, constant row weights $\rho = n^{\frac{l}{l+1}}$, and $ t = (l d +(l+2) \nu + 1) \sqrt[l+1]{n} $ tests.
Furthermore, the decoding can be done in time $\textnormal{poly}(d, \nu) + O(t)$ and each entry can be computed in space $\textnormal{poly}(\log n)$.
\end{theorem}

\section{Low-Energy Massive Random Access}
\label{sec:application}

\begin{figure}
\begin{center}
\scalebox{0.8}{\includegraphics{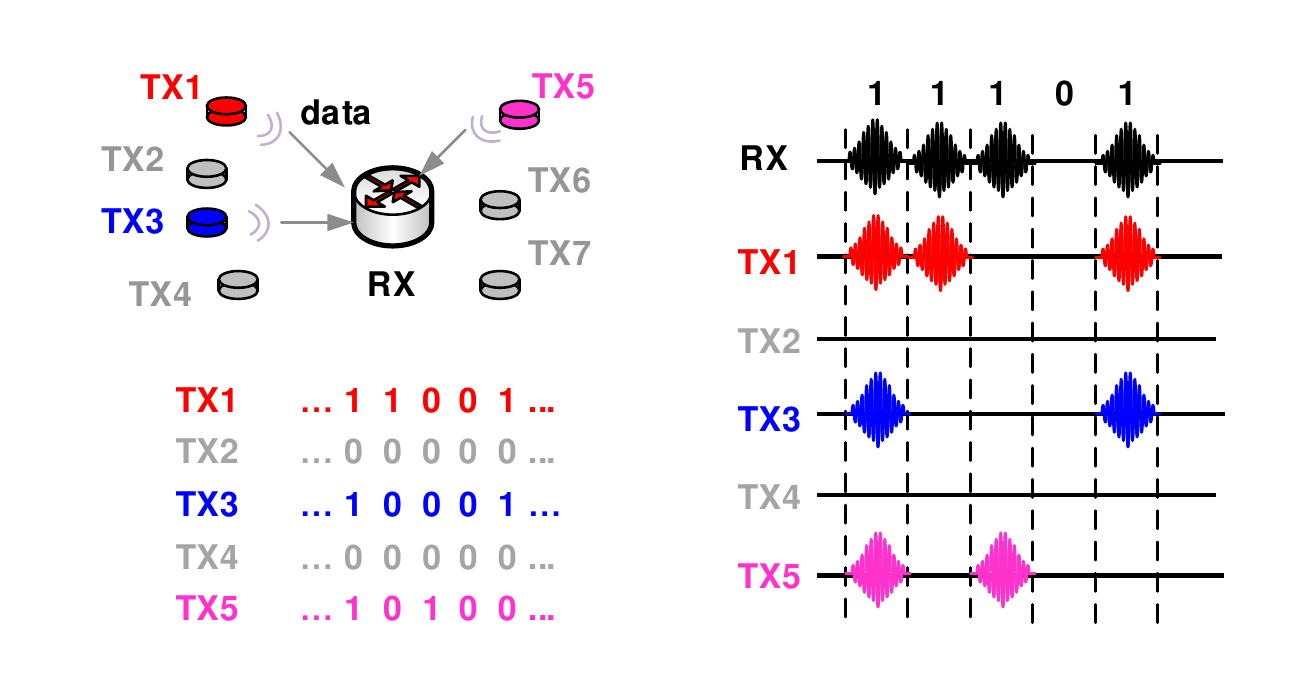}}
\caption{Massive random access with on-off keying at the transmitters and energy detection at the receiver.}
\label{fig:fig1}
\end{center}
\end{figure}

In this section we discuss an application of our framework to wireless random access. Consider  $n$ devices (or sensors) that are associated with a single access point and assume that at most $d$ of them can be active at any given time, where $n\gg d\gg 1$. We adopt the following modulation and detection technique at the transmitters and the receiver respectively: each device uses on-off signaling; i.e., it transmits a binary sequence of 0's and 1's, which corresponds to either transmitting a pulse or no pulse in every time-slot. The access point simply detects whether or not there is energy in the channel in every time-slot. This leads to a (potentially noisy) Boolean OR-channel from the devices to the access point. This simple modulation and detection technique is often used in low-rate applications in practice due to its simplicity. Energy detection does not require any channel state information at the receiver and thus  it eliminates the need for channel training and estimation. This setting is depicted in Figure \ref{fig:fig1}. To simplify the discussion, we focus on the device discovery problem, though as we argue in \cite{inan17allerton} the same group testing  framework can be used to develop solutions for jointly discovering active devices and transmitting data, and for transmitting data without communicating device identities. The device discovery problem can be formulated as follows. Given $n$ devices, design a length-$t$ binary signature for each device (i.e., $M_i\in \{0,1\}^t$ for $i=1,\dots,n$) such that for any set $S\subset\{1,\dots,n\}$ of active devices such that $|S|\leq d$, we can exactly identify the set $S$ (the active devices) from
$$
Y= \left(\bigvee_{i\in S} M_i \right) + v,
$$
where $\bigvee$ denotes entry-wise Boolen-OR operation and '$+$' denotes entry-wise modulo-2 addition, and $v$ is a length-$t$ binary vector representing occasional errors in the energy detection at the receiver (flips of the output) due to noise. We assume that $v$ has at most $\nu$ ones. 

It can be readily observed that the problem statement above corresponds to a (non-adaptive) combinatorial  group testing problem and the binary signature vectors $M_i$ can be taken to be the columns of a $t\times n$ $(d,\nu)$-disjunct group testing matrix. In particular, $(d,\nu)$-disjunct matrices with small $t$ will lead to short binary signatures for the devices. 
In practice, transmitters are subject to energy constraints, which are especially limiting in  IoT applications where devices are required to operate on small batteries for many years or harvest their energy from their environment \cite{HuseyinWioptMAC, HuseyinITTransMAC}. In such applications, while it may  still be desirable to minimize the length of the signature codewords to increase spectral efficiency, it may be also desirable to limit the energy needed to transmit each codeword. The energy spent for transmitting a codeword is proportional to the number of pulses, i.e. the number of ``1''s, in the codeword (ignoring the standby energy for keeping the device active). Using our previous notation, this corresponds to imposing a constraint on the total number of ``1''s in each column of $M$. Note that if energy efficiency were the only metric of interest, we could have resorted to the trivial solution that tests every item individually. This leads to a single ``1'' in each column of $M$ but the length of each column, and therefore that of our signature codewords, becomes equal to $n$. Our sparse group testing framework provides a way to optimally trade energy and spectral efficiency in this framework.

\section{Conclusion \& Discussion}
\label{sec:conclusion}
In this paper, we studied the combinatorial group testing problem under constraints on the number of items that can participate in each test (sparse tests) or the number of tests each item can participate in (sparse codewords). We developed explicit group testing codes that minimize the number of tests under such constraints and proved that they are order optimal or nearly order optimal. Our results show that the minimal number of tests exhibits a particularly favorable behavior in the sparse codewords case, since the number of tests decreases drastically when the number of tests each item  participates in increases beyond a bare minimum.



There are a few remaining gaps in our results which would be interesting to consider in future work. Firstly, as the number of tests per item increases linearly with $d$ (i.e., $w = ld+1$), the gap between our lower and upper bounds on $t$ increases as a function of $d$.  It would be interesting to see if this gap can be closed with sharper lower bounds or improved constructions that yield better performance. Secondly, Kautz and Singleton's construction provides $d$-disjunct matrices with row weight $\rho = n^{\frac{l}{l+1}}$. Therefore, the Kautz and Singleton's construction cannot achieve a row weight of $n^{\alpha}$ for  $\alpha < 1/2$. Nevertheless, as proven in Theorem \ref{random}, $d$-disjunct matrices with $\rho = n^{\alpha}$ for any real number $\alpha \in (0, 1)$ do exist. It would be interesting to know if there are optimal explicit constructions that can achieve $\rho = n^{\alpha}$ for some $\alpha < 1/2$. Finally, while we have exclusively focused on the combinatorial group testing framework in the current paper, where the defective set is to be exactly recovered, we show in \cite{inan18arxiv} that the Kautz and Singleton's construction we consider in this work is also relevant in the probabilistic setting, where the defective set is to be recovered with a small probability of error. In \cite{inan18arxiv}, we build on the Kautz and Singleton's construction to develop the first-order
optimal strongly explicit construction for probabilistic group testing.

\bibliographystyle{IEEEbib}
\bibliography{references}

\appendix

\subsection{Proof of Proposition \ref{prp:sc_noisy_trivial}}
\label{proof:sc_noisy_trivial}

The achievability can trivially be obtained by individual testing, i.e., testing each item alone $\nu + 1$ times. Note that this satisfies $(d, \nu)$-disjunctiveness, therefore, $t(d, \nu, n) \leq (\nu+1) N$.

We can show that for a $t \times n$ binary matrix $M$ that is $(d, \nu)$-disjunct with the condition $w_{\max} \leq d + \nu$, all columns need at least $\nu+1$ private rows, hence $t(d, \nu, n) \geq (\nu + 1) N$. Assume there exists a column $i \in [N]$ with at most $\nu$ private rows. It follows that this column has at least $w_i - \nu$ non-private rows. Fix any $w_i - \nu$ non-private rows. Since $w_{\max} \leq d + \nu$, it follows that $w_i - \nu \leq d$ and we can find at most $d$ other columns covering these rows. Therefore, there exists a set $S$ of columns with $|S| \leq d$ and $i \notin S$ such that $|\textnormal{supp}(M_i) \backslash \cup_{j \in S} \textnormal{supp}(M_j)  | \leq \nu$, which contradicts with $(d, \nu)$-disjunctiveness of $M$. 

\subsection{Proof of Theorem \ref{thm:sc_noisy_first}}
\label{proof:sc_noisy_first}

Let ${M}$ be a $t \times n$ $(d, \nu)$-{disjunct} matrix with $w_{\max} \leq d+\nu+1$. We will separate the columns of $M$ into disjoint groups whose union is $[n]$. We define
\begin{align*}
&\mathcal{N}_1 \coloneqq \left\lbrace j \in [n] \ | \ w_j \leq d+\nu \ \textnormal{or} \ w_j = d+\nu+1 \ \textnormal{and ${M}_j$ has at least $\nu+1$ private rows}
\right\rbrace, \\
&\mathcal{N}_{2, k} \coloneqq \left\lbrace j \in [n] \ | \ w_j = d+\nu+1 \ \textnormal{and ${M}_j$ has $k$ private rows}
\right\rbrace  \ \textnormal{for \ $0 \leq k \leq \nu$}.
\end{align*}
Note that by construction, $\mathcal{N}_1 \cup (\cup_{0 \leq k \leq \nu} \mathcal{N}_{2, k}) = [n]$ and $\mathcal{N}_i \cap \mathcal{N}_j = \emptyset$, hence $n = n_1 + \sum_{0 \leq k \leq \nu} {n}_{2, k}$ where we denote $n_1 \coloneqq \vert \mathcal{N}_1 \vert$ and ${n}_{2, k} \coloneqq \vert \mathcal{N}_{2, k} \vert$ for $0 \leq k \leq \nu$ respectively. In the following, we will bound the size of these sets.

We note that each column in the set $\mathcal{N}_1$ has at least $\nu + 1$ private rows and each column in the set $\mathcal{N}_{2, k}$ has $k$ private rows for $0 \leq k \leq \nu$ respectively. Therefore, we have at least $\alpha \coloneqq (\nu+1) n_1 + \sum \limits_{k=0}^{\nu} k n_{2, k}$ private rows. Since a private row cannot be shared by two distinct columns and there could be at most $t$ private rows, we have
$0 \leq \alpha \leq t$.

Let us now fix a $0 \leq k \leq \nu$ and consider the set $\mathcal{N}_{2, k}$. Take any column $M_i \in \mathcal{N}_{2, k}$ if $ \mathcal{N}_{2, k} \neq \emptyset$. Note that $M_i$ has $k$ private rows and $w_i = d + \nu + 1$. Considering the rest of $d + \nu + 1 - k$ non-private rows, we claim that all size-2 sets of positions of ones must be private. We prove this by contradiction. Assume there exists another column covering any size-2 sets of positions of ones among $d + \nu + 1 - k$ non-private rows. It follows that excluding the covered pair of ones, among the rest of $d + \nu - 1 - k$ non-private rows, one can then find at most $d-1$ other columns covering $d-1$ rows. Therefore, there exists a set $S$ of columns with $|S| \leq d$ and $i \notin S$ such that their union covers $d + 1$ ones of $M_i$. Since $w_i \leq d + \nu + 1$, it follows that $|\textnormal{supp}(M_i) \backslash \cup_{j \in S} \textnormal{supp}(M_j)  | \leq \nu$ which contradicts with $(d, \nu)$-disjunctiveness. It follows that there are $\binom{d + \nu + 1 - k}{2}$ private size-2 sets of positions of ones for any $M_i \in \mathcal{N}_{2, k}$. By definition of a private set it cannot be shared by two distinct columns and we excluded the private rows in our calculation for the number of private size-2 sets, hence we have $\sum_{0 \leq k \leq \nu} n_{2, k} \binom{d +\nu+1-k}{2}$ private size-2 sets whereas there could be at most $\binom{t - \alpha}{2}$ private size-2 sets. Therefore, we have
\begin{align*}
\sum_{0 \leq k \leq \nu} n_{2, k} \binom{d+\nu+1-k}{2} \leq \binom{t-\alpha}{2} \leq \dfrac{(t - \alpha)^2}{2}.
\end{align*}
Hence, this gives
\begin{align*}
t & \geq  \sqrt{\sum_{0 \leq k \leq \nu} n_{2, k} (d+\nu+1-k)(d+\nu-k)} + \alpha
\\ & = \sqrt{\sum_{0 \leq k \leq \nu} n_{2, k} (d+\nu+1-k)(d+\nu-k)} + (\nu+1) n_1 + \sum \limits_{k=0}^{\nu} k n_{2, k}
\end{align*}
with the condition $n = n_1 + \sum_{0 \leq k \leq e} n_{2, k}$. We can also write this as 
\begin{align*}
t \geq \sqrt{\sum_{0 \leq k \leq \nu} n_{2, k} (d+\nu+1-k)(d+\nu-k)} + (\nu+1) \left(n - \sum_{0 \leq k \leq \nu} n_{2, k} \right)+ \sum \limits_{k=0}^{\nu} k n_{2, k}
\end{align*}
with the condition $0 \leq \sum_{0 \leq k \leq \nu} n_{2, k} \leq n$.

Since this is a concave function over $\{n_{2, k}\}_{k=0}^{\nu}$ and $0 \leq \sum_{0 \leq k \leq \nu} n_{2, k} \leq n$ is a convex set, minimum is attained over one of the extreme points. Therefore, we obtain
\begin{align*}
t & \geq \min \left\{ (\nu+1)n, \min_{0\leq k \leq \nu} \sqrt{(d + \nu +1 - k)(d + \nu - k)n} + n k \right\} \\ & = \min \left\{ (\nu+1)n, \sqrt{(d + \nu)(d + \nu + 1)n} \right\}.
\end{align*}

\subsection{Proof of Theorem \ref{thm:sc_noisy_general}}
\label{proof:sc_noisy_general}

Let ${M}$ be a $t \times n$ $(d, \nu)$-{disjunct} matrix $w_{\max} \leq ld+\nu+1$. We define
\begin{align*}
&\mathcal{N}_1 \coloneqq \left\lbrace j \in [n] \ | \ w_j \leq d + \nu \ \textnormal{or} \ w_j = d + \nu + 1 \ \textnormal{and ${M}_j$ has at least one private row}
\right\rbrace, \\
&\mathcal{N}_i \coloneqq \lbrace j \in [n] \ | \  (i-2)d + \nu + 2\leq w_j \leq (i-1)d + \nu + 1 \ \textnormal{and ${M}_j$ has no private set of size $i-1$} \
\\ &  \ \ \ \ \ \ \ \ \ \ \ \ \ \ \
\textnormal{or}
\\ &  \ \ \ \ \ \ \ \ \ \ \ \ \ \ \
(i-1)d + \nu + 2 \leq w_j \leq i d + \nu + 1 \ \textnormal{and ${M}_j$ has at least one private set of size $i$} \rbrace,
\\ & \textnormal{for $i = 2, \ldots, l$},
\\ & \mathcal{N}_{l+1} \coloneqq \lbrace j \in [n] \ | (l-1)d + \nu + 2\leq w_j \leq l d + \nu + 1 \ \textnormal{and ${M}_j$ has no private set of size $l$} \rbrace.
\end{align*}
Note that by construction, $\mathcal{N}_1 \cup \ldots \cup \mathcal{N}_{l+1} = [n]$ and $\mathcal{N}_i \cap  \mathcal{N}_j = \emptyset$ for any $i, j \in [l+1]$ such that $i \neq j$, hence $n = |\mathcal{N}_1| + \ldots + |\mathcal{N}_{l+1}|$. In the following, we will bound the size of these sets.

Note that $|\mathcal{N}_1| \leq t$. Consider the sets $\mathcal{N}_i$ for $i = 2, \ldots, l$. For any column $j \in \mathcal{N}_i$, if we have $(i-1)d + \nu + 2 \leq w_j \leq i d + \nu + 1$, then by construction ${M}_j$ has at least one private set of size $i$. For the case $(i-2)d + \nu + 2\leq w_j \leq (i-1)d + \nu + 1$, using similar arguments as in the proof of Theorem \ref{thm:sparsecodewords} one can show that all the sets of positions of ones of size $i$ must be private for the column ${M}_j$. Hence, all the columns in the set ${N}_i$ must have at least one private set of size $i$. Since the private sets cannot be shared among columns and we have at most $\binom{t}{i}$ private sets of size $i$, it yields that $|\mathcal{N}_i| \leq \binom{t}{i}$.  For the last set $\mathcal{N}_{l+1}$, similar arguments apply and for each column all the set of positions of ones of size $l+1$ must be private. Since $w_j \geq (l-1)d + \nu + 2$ for $j \in \mathcal{N}_{l+1}$, we have $|\mathcal{N}_{l+1}| \binom{(l-1)d + \nu + 2}{l+1} \leq \binom{t}{l+1}$. Therefore,
\begin{align*}
n & = |\mathcal{N}_1| + \ldots + |\mathcal{N}_{l+1}| \\ & \leq \sum \limits_{i=1}^{l} \binom{t}{i}  + \dfrac{\binom{t}{l+1}}{\binom{(l-1)d + \nu + 2}{l+1}}
\\ & \overset{(i)}{\leq}
\left( \dfrac{{e} t}{l} \right)^l + \dfrac{t \ldots (t - l)}{ ((l-1)d + \nu + 2) \ldots ((l-1)(d-1) + \nu + 1)}
\\ & \overset{(ii)}{\leq}
\dfrac{{e}^l t^l}{l^l}  + \dfrac{t^{l+1}}{ ((l-1)(d-1) + \nu)^{l+1} }
\\ & \overset{(iii)}{\leq}
\dfrac{{e}^l t^{l}}{(l-1)^l} \dfrac{2t}{(d+v)^2}  + \dfrac{t^{l+1}}{ ((l-1)(d-1) + \nu)^{l+1} }
\\ & = t^{l+1} \left( \dfrac{2 e^l}{(d+v)^2 (l-1)^l} + \dfrac{1}{((l-1)(d-1) + \nu)^{l+1}} \right)
\end{align*}
where $(i)$ is due to the inequality $\sum \limits_{i=0}^{l} \binom{t}{i} \leq \left( \frac{e t}{l} \right)^l$
for $t \geq l \geq 1$, $(ii)$ is bounding all the terms in the numerator by $t$ and denominator by $(l-1)(d-1) + \nu$ and in $(iii)$  we use $t \geq \binom{d+v+2}{2} \geq \frac{(d+v)^2}{2}$. This completes the lower bound.

\subsection{Proof of Theorem \ref{thm:decoding_noisy}}
\label{proof:decoding_noisy}

The decoding procedure follows what is described in the proof of Theorem \ref{thm:decoding}.  For each $i \in [t_q]$, we create the sets $S_i \subseteq [q]$ such that $S_i$ consists of the set of position of ones in $Y_i$. Due to the noise, we can only guarantee that $\lvert S_i \rvert \leq d + \nu$, for every $i \in [t_q]$ in this case. Similarly, for any defective item, the corresponding codeword $(c_1, \ldots, c_{t_q})$ in the outer code must satisfy $|\{i : c_i \in S_i \}| \geq t_q - \nu$ and for any non-defective item, the corresponding codeword $(c_1, \ldots, c_{t_q})$ in the outer code will include at least $(l+1)\nu \geq 2\nu$ symbols $c_i$ such that $c_i \notin S_i$. Note that this step can be done in $O(t)$ time.

The second step is to output all codewords $(c_1, \ldots, c_{t_q})$ in the outer code such that
$|\{i : c_i \in S_i \}| \geq t_q - \nu$ given $S_i \subseteq [q]$ with $\lvert S_i \rvert \leq d + \nu$ for every $1 \leq i \leq t_q$. This problem is an instance of the list recovery problem \cite{rudra_lr, indyk2001, zuckerman2004} and it can be solved in polynomial time for a $[t_q, k_q, t_q-k_q+1]_q$ Reed-Solomon code as long as $t_q - \nu > \sqrt{(k_q - 1) (d+\nu) t_q}$. We note that in our case, we have $t_q = l d + (l+2) \nu + 1$, and $k_q = l+1$ which satisfies the requirement.

\end{document}